\documentclass[graybox, natbib, nosecnum]{svmult}
\bibpunct{(}{)}{;}{a}{}{,} % suppress commas between author-names and year

\usepackage{mathptmx}       % selects Times Roman as basic font
\usepackage{helvet}         % selects Helvetica as sans-serif font
\usepackage{courier}        % selects Courier as typewriter font
\usepackage{type1cm}        % activate if the above 3 fonts are
\usepackage{graphicx}        % standard LaTeX graphics tool
\usepackage{multicol}        % used for the two-column index
\usepackage[bottom]{footmisc}% places footnotes at page bottom
\usepackage[normalem]{ulem}	% for strike-through of text with \sout{}  
\usepackage{hyperref}  %for hyperlinks
\usepackage{soul}   % for high-lighting of text
\usepackage{amsmath}
\usepackage{latexsym}
\usepackage{amssymb}
\usepackage{amsfonts}
\usepackage{mathtools}
\usepackage{bm}
\usepackage{algorithm}
\usepackage{algpseudocode}

\newcommand{\tr}{\textrm{Tr}}
\newcommand{\rs}{\textrm{RS}}
\newcommand{\grm}{\textrm{RM}}

\def\F{\Bbb F}

\newcommand{\rmv}[1]{}

\begin{document}
%\tableofcontents{}
\title*{A repair scheme for a distributed storage system based on multivariate polynomials}
%Unification of repair schemes for linear codes in distributed storage and computation
\author{Hiram H. L\'opez\thanks{\hskip -0.15 cm Corresponding author. Hiram H. L\'opez was partially supported by NSF DMS-2401558. Gretchen L. Matthews was partially supported by NSF DMS-2201075 and the Commonwealth Cyber Initiative.}, Gretchen L. Matthews, and Daniel Valvo}
\institute{Hiram H. L\'opez \at Department of Mathematics, Virginia Tech, Blacksburg, VA USA \email{hhlopez@vt.edu},
Gretchen L. Matthews \at Department of Mathematics, Virginia Tech, Blacksburg, VA USA \email{gmatthews@vt.edu}, and
Daniel Valvo \at Department of Mathematics, Virginia Tech, Blacksburg, VA USA \email{vdaniel1@vt.edu}.}
\maketitle
\abstract{
A distributed storage system stores data across multiple nodes, with the primary objective of enabling efficient data recovery even in the event of node failures. The main goal of an exact repair scheme is to recover the data from a failed node by accessing and downloading information from the rest of the nodes. In a groundbreaking paper, ~\cite{GW} developed an exact repair scheme for a distributed storage system that is based on Reed-Solomon codes, which depend on single-variable polynomials. In these notes, we extend the repair scheme to the family of distributed storage systems based on Reed-Muller codes, which are linear codes based on multivariate polynomials. The repair scheme we propose repairs any single node failure and multiple node failures, provided the positions satisfy certain conditions.}

\section{Keywords}
Reed-Solomon code, Reed-Muller code, repair scheme, distributed storage system

%%%%%%%%%%%%%%%%%%%%%%%%%%%%%%%%%%%%%%%%%%%%%%%%%%%%%%%%%%%%%%%%%%%%%%%%%%%%%%%%
\section{Introduction}
The goal of a distributed storage system is to store data over multiple storage nodes. A linear code, which is a vector space over a finite field, may be used in a distributed storage system setting to allow the information stored on a failed node to be recovered using the information stored on the remaining nodes. The general idea is the following.
\begin{itemize}
\item The information to be stored is encoded into codewords using a linear code.
\item Every codeword is distributed across nodes so that each node stores a symbol.
\item Recovering a failed node exactly is equivalent to fixing an erasure in the codeword \cite{DSS1}, \cite{DSS2}.
\end{itemize}
A repair scheme is an algorithm that recovers the value at any node using limited information from the values at the other nodes. Under certain conditions, some repair schemes require less information than standard approaches to repair. Thus, the mathematical goal is the following.
\begin{itemize}
\item Design a set of vectors in such a way that every entry of every vector can be recovered from the rest of the entries. In these notes, we use evaluation codes, meaning that the set of vectors is a vector space over a finite field, and every vector depends on the evaluation of a certain polynomial.
\item Give an explicit description of an exact repair scheme. In other words, give the algorithm describing how every vector entry can be recovered from the rest of the entries. In these notes, we use the trace function from finite fields to describe the repair scheme.
\end{itemize}

An evaluation code is a linear code defined by evaluating a collection of polynomials on a set of points. Reed-Solomon codes,  the most well-known family of evaluation codes, are defined by the evaluation of single-variable polynomials up to a certain degree on a set of points of the finite field $\mathbb{F}_q$ of size $q$. The design of linear exact repair schemes for distributed storage systems using evaluation codes began with the foundational work of Guruswami and Wootters in which they developed a repair scheme (GW scheme) to efficiently repair an erasure in a Reed-Solomon (RS) code; see~\cite{GW}. The GW scheme highly depends on the dual of a Reed-Solomon code, which is a generalized Reed-Solomon code. For a general framework for evaluation codes, see~\cite{JVV}. For the dual of an evaluation code, see~\cite{Lopez2021}.

The GW scheme inspired the linear exact repair schemes for algebraic geometry codes \cite{JLX} and Reed-Muller codes~\cite{RM}. Similarly to Reed-Solomon codes, the Reed-Muller codes are defined by evaluating polynomials up to a certain degree in $m$ variables on the points $\mathbb{F}_{q^t}^m$. 

In these notes, we develop a repair scheme for several failures on a distributed storage system that is based on Reed-Muller codes, provided the positions satisfy a certain restriction. The approach we develop in these notes, which relies on the dual of an evaluation code~\cite{Lopez2021}, is different from the one used in~\cite{RM}, and it gives the basis to extend it to other families of codes, for instance, the family of Cartesian codes~\cite{Cartesian-codes}.

%%%%%%%%%%%%%%%%%%%%%%%%%%%%%%%%%%%%%%%%%%%%%%%%%%%%%%%%%%%%%%%%%%%%%%%%%%%%%%%%
\section{Preliminaries}
Let $q$ be a power of a prime $p$ and $\mathbb{F}_{q^t}$ be a field extension of degree $t = [\mathbb{F}_{q^t} : \mathbb{F}_q]$ of $\mathbb{F}_q.$ Let $C$ be an $[n,k]$-linear code over $\mathbb{F}_{q^t}$, meaning a $k$-dimensional $\mathbb{F}_{q^t}$-subspace of $\mathbb{F}_{q^t}^n$. The elements of $\mathbb{F}_{q^t}$ are called {\it symbols} and the elements of $\mathbb{F}_q$ are called {\it subsymbols}. As $\mathbb{F}_{q^t}$ is a vector space of dimension $t$ over $\mathbb{F}_{q}$, any codeword $c\in C$ (or more generally any element $w \in \mathbb{F}_{q^t}^n$) can be represented using $n$ symbols or $tn$ subsymbols.

\subsection{Field trace}
The {\it field trace} can be defined as the polynomial  $\tr_{{\mathbb{F}_{q^t}}/\mathbb{F}_q}(x) \in {\mathbb{F}_{q^t}}[x]$ given by
\[\tr_{{\mathbb{F}_{q^t}}/\mathbb{F}_q}(x) := x^{q^{t-1}} + \dots + x^{q}+x.\]
For the sake of convenience, we will often refer to $\tr_{{\mathbb{F}_{q^t}}/\mathbb{F}_q}(x)$ as simply $\tr(x)$ when the extension being used is obvious from context.

The following property of the trace function is crucial to developing the repair scheme.

\begin{remark}\label{21.01.02}
Let $B = \{z_1, \dots, z_t\}$ be a basis for $\mathbb{F}_{q^t}$ over $\mathbb{F}_q$. Then there exists a dual basis $B^\prime = \{z^\prime_1, \dots, z^\prime_t\}$ of $\mathbb{F}_{q^t}$ over $\mathbb{F}_q$, such that $$\tr(z_i z^\prime_j)=\delta_{ij}:=
\begin{cases}
    1 & i=j\\
    0 & \textnormal{otherwise}.
\end{cases}
$$ In this case, $B$ and $B^\prime$ are called {\it dual bases}. For $\alpha \in \mathbb{F}_{q^t}$, \[\alpha = \sum_{i=1}^t \tr(\alpha z_i)z_i^\prime.\]
Thus, given dual bases $B$ and $B'$, determining $\alpha$ is equivalent to finding $\tr(\alpha z_i)$ for all $i \in [t]:=\{ 1, \dots, t \}$; see \cite{Lidl_Niederreiter_1994}.
\end{remark}

\subsection{Reed-Muller codes}
By Reed-Muller codes, we mean the evaluation codes obtained when polynomials in $m$ variables up to a certain total degree $d\in \mathbb{Z}_{\geq 0}$ are evaluated at all the points of $\mathbb{F}_{q^t}^m$. The precise definition is as follows.

Let $\mathbb{F}_{q^t}[x_1, \dots, x_m]$ be the polynomial ring of $m$ variables over $\mathbb{F}_{q^t}.$
Denote by $\mathbb{F}_{q^t}[x_1, \dots, x_m]_{\leq d}$ the set of those polynomials up to a certain total degree $d$.

\begin{definition}
Assume $\{P_1,\ldots,P_n \} = \mathbb{F}_{q^t}^m$, fixing an order on the $n:={q^t}^m$ elements of $\mathbb{F}_{q^t}^m$ . The \textit{Reed-Muller code} of degree $d$ is given by
\[\grm(\mathbb{F}_{q^t}^m, d) := \left\{ f(\mathbb{F}_{q^t}^m) \ : \ f \in \mathbb{F}_{q^t}[x_1, \dots, x_m]_{\leq d} \right\} \subseteq \mathbb{F}_{q^t}^n,\]
where $f(\mathbb{F}_{q^t}^m) := (f(P_1), \dots, f(P_n))$.    
\end{definition}

\begin{definition}\label{24.02.12}
Note that a Reed-Solomon code of dimension $k$ of length $q^t$ is defined as
\[\rs(\mathbb{F}_{q}^m, k) := \grm(\mathbb{F}_{q^t}, k-1).\]
\end{definition}

The dual of a Reed-Muller code, denoted by $\grm(\mathbb{F}_{q^t}^m, d)^\perp$, is the set of all $\alpha\in \mathbb{F}_{q^t}^n$ such that $\alpha \cdot \beta=0$ for all $\beta\in \grm(\mathbb{F}_{q^t}^m, d)$, where $\alpha \cdot \beta$ is the ordinary inner product in $\mathbb{F}_{q^t}^n$. The dual code $\grm(\mathbb{F}_{q^t}^m, d)^\perp$ has been extensively studied in the literature. See, for instance,~\cite{DGM} and ~\cite{Huffman-Pless}, where it is shown that the dual of a $\grm$-code is another $\grm$-code:
\[\grm(\mathbb{F}_{q^t}^m, d)^\perp=\grm(\mathbb{F}_{q^t}^m, d^\perp),\]
where $d^\perp := m(|\mathbb{F}_{q^t}| - 1) - d - 1 = m(q^t - 1) - d - 1$.

\subsection{Exact repair scheme}
In terms of distributed storage systems, each entry $c_i$ of a codeword $c \in C$ represents the information stored on one of $n$ different storage nodes. Informally, when one of the storage nodes fails, meaning that it is unavailable to serve a data request, an exact repair scheme is an algorithm designed to recover the information of the erased node in terms of data held by the other storage nodes. Formally, we have the following definition.
\begin{definition}
Let $C$ be a linear code over $\mathbb{F}_{q^t}$ of length $n$ and dimension $k$, given by a collection of functions $\mathcal{F}$ and a set of evaluation points $A$. A linear exact repair scheme for $C$ over $\mathbb{F}_{q}$ consists of the following.
\begin{itemize}
\item For each $\alpha^* \in A$, and for each $\alpha \in A \setminus \{\alpha^*\}$, a set of queries $Q_\alpha(\alpha^*)\subseteq \mathbb{F}_{q^t}$.
\item For each $\alpha^* \in A$, a linear reconstruction algorithm that computes
\[f(\alpha^*) = \sum_i \lambda_i z_i \]
for coefficients $\lambda_i \in B$ and a basis $\{z_1, \dots, z_t\}$ for $\mathbb{F}_{q^t}$ over $\mathbb{F}_q$ so that the coefficients $\lambda_i$ are $\mathbb{F}_q$-linear combinations of the queries
\[\bigcup_{\alpha \in A \setminus \{\alpha^*\}} \left\{ \tr_{{\mathbb{F}_{q^t}}/\mathbb{F}_q}(\gamma f(\alpha)) : \gamma \in Q_\alpha(\alpha^*) \right\}.\]
\end{itemize}
\end{definition}
We often omit the word linear because we consider only linear exact repair schemes.

The {\it repair bandwidth} $b$ is the number of subsymbols the scheme downloads to recover the erased entry. As any element $c \in \mathbb{F}_{q^t}^n$ depends on $nt$ subsymbols, the number $\frac{b}{nt} $ can be seen as the fraction of the codeword that is needed by the exact repair scheme to recover the erased symbol. It is important to note that to recover a symbol $c_i \in \mathbb{F}_{q^t}$ of a codeword $c$, these $b$ elements of $\mathbb{F}_q$ rely on the entries $c_j$, $j \neq i$, of $c$ but are not necessarily $b$ of them. 
%%%%%%%%%%%%%%%%%%%%%%%%%%%%%%%%%%%%%%%%%%%%%%%%%%%%%%%%%%%%%%%%%%%%%%%%%%%%%%%%
\section{A one-erasure repair scheme of a Reed-Muller code}
In this section, we adapt the GW scheme of one erasure from Reed-Solomon codes to Reed-Muller codes.

\begin{remark}
In \cite{9624941}, the authors developed a repair scheme of one erasure for certain decreasing monomial-Cartesian and augmented Reed-Muller codes. We decided to focus these notes on Reed-Muller codes to determine, in Theorem~\ref{24.02.11}, to what degree a Reed-Muller code can be repaired using ideas similar to~\cite{GW}.
\end{remark}

\begin{theorem}\label{24.02.11}
Let $\grm(\mathbb{F}_{q^t}^m, d)$ be a Reed-Muller code such that $d\leq m(q^t-1)-q^{t-1}$. There exists a repair scheme for one erasure with bandwidth at most
%\[b = q^t-1 + (t-1)\left(\frac{q^{t(m-1)}}{q^t}-1\right). \]
\[b = q^{mt}-1 + (t-1)\left(q^{(m-1)t}-1\right). \]
\end{theorem}
\begin{proof}
Let $c:=(f(P_1),\ldots,f(P_n)) \in \grm(\mathbb{F}_{q^t}^m, d)$.  Assume that the entry of $c$ corresponding to $P^*=(p_1,\ldots,p_m) \in \mathbb{F}_{q^t}^m$, meaning $f(P^*)$, has been erased.

Recall that $\{z_1, \dots, z_t\}$ is a basis for $\mathbb{F}_{q^t}$ over $\mathbb{F}_q$. For $i\in [t]$, define the following $t$ polynomials, which we refer to as recovery polynomials:
\begin{align*}
r_i(x) &=
\frac{\tr(z_i(x_1 - p_1))}{(x_1 - p_1)} \\
&= {z_i} +
z_i^q(x_1 - p_1)^{q-1}+ \cdots +
z_i^{q^{t-1}}(x_1 - p_1)^{q^{t-1}-1}.
\end{align*}

As $d\leq m(q^t-1)-q^{t-1}$, then
\begin{align*}
d^\perp &= m(q^t - 1) - d - 1\\
& \geq m(q^t - 1) - (m(q^t-1)-q^{t-1}) -1\\
& = q^{t-1} -1 \geq \deg r_i(x).
\end{align*}
It follows that every polynomial $r_i(x)$ defines an element in $\grm(\mathbb{F}_{q^t}^m, d)^\perp$, 
meaning  $r_i(\mathbb{F}_{q^t}^m) \in \grm(\mathbb{F}_{q^t}^m, d)^\perp$.

As a consequence, we obtain the $t$ equations
\begin{equation}\label{21.06.15}
r_{i}(P^*)f(P^*)= -\sum_{\mathbb{F}_{q^t}^m \setminus\{P^*\}}
 r_{i}(P)f(P)  \quad \forall i\in[t].
\end{equation}
As $r_i(P^*)=z_i,$ applying the trace function to both sides of the previous equations and employing the linearity of the trace function, we obtain
\[\tr \left(z_i  f(P^*)\right)= 
-\sum_{\mathbb{F}_{q^t}^m \setminus\{P^*\}}
\tr \left(  r_{i}(P)f(P) \right) \quad \forall i\in[t].\]

Define the set $ \Gamma := \{p_1\} \times \mathbb{F}_{q^t} \times \cdots \times \mathbb{F}_{q^t} \subset \mathbb{F}_{q^t}^m.$ Then,
$$
r_i(P)= \begin{cases}
   z_i & \textnormal{if }  P \in \Gamma \\ \ \\
 \frac{ \tr(z_i(s_P - p_1))}{(s_P - p_1)}   & \textnormal{if }   {P \notin \Gamma},
\end{cases}
$$
where $s_P$ is the first entry of the point $P \in \mathbb{F}_{q^t}^m$. Note $P^* \in \Gamma$.
Therefore, we obtain that for $i\in[t],$
$$
\begin{array}{ll}
\tr \left(z_i  f(P^*)\right)&=\sum_{\mathbb{F}_{q^t}^m \setminus\{P^*\}} \tr \left(  r_{i}(P)f(P) \right)\\ \ \\
&= \sum_{\Gamma \setminus\{P^*\}} \tr \left(  r_{i}(P)f(P) \right) + 
\sum_{\mathbb{F}_{q^t}^m \setminus \Gamma} \tr \left(  r_{i}(P)f(P) \right) \\ \ \\
& = \sum_{\Gamma \setminus\{P^*\}} \tr \left(  z_i f(P) \right) + 
\sum_{\mathbb{F}_{q^t}^m \setminus \Gamma} \tr \left(  \frac{\tr(z_i(s_P - p_1))}{s_P - p_1}f(P) \right)\\ \ \\
& = \sum_{\Gamma \setminus\{P^*\}} \tr \left(  z_i f(P) \right) + 
\sum_{\mathbb{F}_{q^t}^m \setminus \Gamma} \tr(z_i(s_P {-} P_1)) \tr \left( \frac{ f(P)}{s_P {-} p_1} \right).
\end{array}
$$
The properties of the trace function imply that the entry $f(P^*)$ can be recovered from its $t$ independent traces $\tr (z_i  f(P^*)),$ which can be obtained by downloading the following information from every element $P \neq P^*$:
\begin{itemize}
\item if $P \in \Gamma \setminus\{P^*\}$,  download $f(P)$. 
\item if $P \notin \Gamma$, { download } $\displaystyle \tr\left(\frac{ f(P)}{ s_P - p_1}\right)$. 
\end{itemize}
Hence, the bandwidth is 
\begin{align*}
b & = t (|\Gamma|-1) + |\mathbb{F}_{q^t}^m\setminus \Gamma| = t (q^{(m-1)t}-1) + q^{mt} - q^{(m-1)t}\\
& = tq^{(m-1)t} - t + q^{mt} - q^{(m-1)t}\\
& = (t-1)q^{(m-1)t} - t + q^{mt}\\
&= q^{mt}-1 + (t-1)\left(q^{(m-1)t}-1\right),
\end{align*}
which concludes the proof.
\end{proof}

As a consequence of Theorem~\ref{24.02.11}, when $m=1$, we obtain the GW repair scheme for Reed-Solomon codes.
\begin{corollary}
Let $\grm(\mathbb{F}_{q^t}, d)$ be a Reed-Muller code such that $d\leq q^t-q^{t-1}-1$, meaning a Reed-Solomon code. Then, there exists a repair scheme for one erasure with bandwidth at most
\[b = q^t-1. \]
\end{corollary}
\begin{proof}
By Definition~\ref{24.02.12}, $\rs(\mathbb{F}_{q}^m, k) = \grm(\mathbb{F}_{q^t}, k-1)$. So, we obtain the result as a consequence of Theorem~\ref{24.02.11}.
\end{proof}

%%%%%%%%%%%%%%%%%%%%%%%%%%%%%%%%%%%%%%%%%%%%%%%%%%%%%%%%%%%%%%%%%%%%%%%%%%%%%%%%
\section{A two-erasures repair scheme of a Reed-Muller code}
In this section, we adapt the GW scheme of one erasure from Reed-Solomon codes to two erasures on a Reed-Muller code.

We keep the same notation as in previous sections and develop a repair scheme that repairs two simultaneous erasures $f(\bm{s^{\prime}})$ and $f(\bm{s}^*)$ on a distributed storage system based on a Reed-Muller code $\grm(\mathbb{F}_{q^t}, d)$ provided the erasure positions satisfy a certain condition.

\begin{remark}
In \cite{9624941}, the authors developed a repair scheme of two erasures for certain decreasing monomial-Cartesian and augmented Reed-Muller codes. We focused these notes on Reed-Muller codes to determine in Theorem~\ref{24.02.11} to which degree a Reed-Muller code can be repaired using ideas similar to~\cite{GW}.
\end{remark}

\begin{theorem}\label{21.06.16}
Let $c = (f(\bm{s}_1),\ldots,f(\bm{s}_n)) \in \grm(\mathbb{F}_{q^t}^m, d)$, where $\grm(\mathbb{F}_{q^t}^m, d)$ is a Reed-Muller code with $d\leq m(q^t-1)-q^{t-1}$ and $n=q^{tm}$. Assume that $c$ has the two erasures $f(\bm{s^{\prime}})$ and $f(\bm{s}^*)$, where $\bm{s}^\prime = \left(s_1^\prime,\ldots,s_n^\prime \right)$ and $\bm{s}^* = \left(s_1^*,\ldots, s_n^* \right)$. If there is $j \in n$ such that $s^{\prime}_j - s^*_j\in \mathbb{F}_q^*$, then there exists a repair scheme for the two erasures with bandwidth at most
\[b = 2\left[ n-2 + (t-1)\left(q^{(m-1)t}-2\right)\right].\]
\end{theorem}
\begin{proof}
Note that the kernel of the trace $\ker \tr =  \left\{\alpha \in \F_{q^t} : \tr \left(\alpha   \right) = 0 \right\}$ has dimension $t-1$ as an $\mathbb{F}_q$-vector space. Let $\{z_1, \dots, z_{t - 1}\}$ be an $\mathbb{F}_q$-basis for $\ker \tr$ and $z_t$ an element in $\F_{q^t}$ such that $\{z_1, \dots, z_{t - 1}, z_t\}$ is an $\mathbb{F}_q$-basis for $\F_{q^t}.$ Finally, let $\tau \in \ker \tr.$ % and define $\tau =\displaystyle \frac{\overline{\tau}}{\lambda_{\bm{s^{\prime}}}}.$ %$\Gamma_{\bm{s}^*}^{\bm{s^{\prime}}}.$ Define $\overline{\tau} = \tau\prod_{i=1}^m (s^{\prime}_i - s^*_i ).$
We are ready to define the repair polynomials. For $i \in [t]$, take
\[p_i(\bm{x}) = \tau\frac{\tr \left(z_i (x_j - s^*_j )\right)}{ \left(x_j - s^*_j \right)}
\qquad \text{ and } \qquad
q_i(\bm{x}) = \frac{\tr\left(z_i (x_j - s^{\prime}_j )\right)}{ (x_j - s^{\prime}_j )}.\]
As $d\leq m(q^t-1)-q^{t-1}$, the polynomials $p_i(\bm{x})$ and $q_i(\bm{x})$ define elements in the dual code $\grm(\mathbb{F}_{q^t}^m, d)^\perp$. Therefore, in a similar way to the proof of Theorem~\ref{24.02.11}, we obtain the following $2t$ equations:
\begin{align}
&p_{i}(\bm{s}^*)f(\bm{s}^*) + p_{i}(\bm{s^{\prime}})f(\bm{s^{\prime}}) =-\sum_{\bm{s} \in \mathbb{F}_{q^t}^m \setminus\{\bm{s}^*, \bm{s^{\prime}}\}}  p_{i}(\bm{s})f(\bm{s}), \qquad i\in[t]\label{Eq.21.06.11_1}
\end{align}
and
\begin{align}
&q_{i}(\bm{s}^*)f(\bm{s}^*) + q_{i}(\bm{s^{\prime}})f(\bm{s^{\prime}}) =-\sum_{\bm{s} \in \mathbb{F}_{q^t}^m \setminus\{\bm{s}^*, \bm{s^{\prime}}\}} q_{i}(\bm{s})f(\bm{s}),  \qquad i\in[t].\label{Eq.21.06.11_2}
\end{align} 
By definition of the $p_i$'s, we have
\begin{align*}
& p_{i}(\bm{s}^*) = \tau z_i \qquad \text{and} \qquad p_{i}(\bm{s^{\prime}}) = \tau \tr\left(z_i \right), \qquad i\in[t].
\end{align*}
By definition of the $q_i$'s, we have
\begin{align*}
& q_{i}(\bm{s}^*) = \tr\left(z_i \right) \qquad \text{and} \qquad q_{i}(\bm{s^{\prime}}) = z_i, \qquad i\in[t].
\end{align*}
Equations~\ref{Eq.21.06.11_1} and~\ref{Eq.21.06.11_2} become
\begin{align*}
&\tau z_if(\bm{s}^*) + \tau \tr\left(z_i \right) f(\bm{s^{\prime}}) =-\sum_{\bm{s} \in \mathbb{F}_{q^t}^m \setminus\{\bm{s}^*, \bm{s^{\prime}}\}}  p_{i}(\bm{s})f(\bm{s}), \qquad i\in[t]
\end{align*}
and
\begin{align*}
&\tr\left(z_i\right) f(\bm{s}^*) + z_i f(\bm{s^{\prime}}) =-\sum_{\bm{s} \in \mathbb{F}_{q^t}^m \setminus\{\bm{s}^*, \bm{s^{\prime}}\}} q_{i}(\bm{s})f(\bm{s}),  \qquad i\in[t].
\end{align*}
As $\{z_1, \dots, z_{t - 1}\}$ is an $\mathbb{F}_q$-basis for $\ker \tr$, the last two equations imply
\begin{align}
\tr \left(\tau z_i f(\bm{s}^*)\right)
&=\label{Eq.21.06.11_3}
-\sum_{\bm{s} \in \mathbb{F}_{q^t}^m \setminus\{\bm{s}^*, \bm{s^{\prime}}\}} \tr \left( p_{i}(\bm{s})f(\bm{s})\right),\quad i\in[t-1],\\
\tr \left(\tau z_t f(\bm{s}^*)\right) + \tr \left(z_t \right) \tr \left(\tau f(\bm{s^{\prime}})\right)
&=-\sum_{\bm{s} \in \mathbb{F}_{q^t}^m \setminus\{\bm{s}^*, \bm{s^{\prime}}\}} \tr \left( p_{t}(\bm{s})f(\bm{s})\right), \label{Eq.21.06.11_4}\\
\tr \left( z_i f(\bm{s^{\prime}}) \right)
&=\label{Eq.21.06.11_5}
-\sum_{\bm{s} \in \mathbb{F}_{q^t}^m \setminus\{\bm{s}^*, \bm{s^{\prime}}\}} \tr \left( q_{i}(\bm{s})f(\bm{s})\right) ,\quad i\in[t-1],\\
\tr \left(z_t \right) \tr \left(f(\bm{s}^*)\right)+ \tr \left(z_t f(\bm{s^{\prime}})\right)
&= -\sum_{\bm{s} \in \mathbb{F}_{q^t}^m \setminus \{\bm{s}^*, \bm{s^{\prime}}\}} \tr \left( q_{t}(\bm{s})f(\bm{s})\right). \label{Eq.21.06.11_6}
\end{align}
We claim that the elements $f(\bm{s}^*)$ and $f(\bm{s}^\prime)$ can be recovered from the set \[ R:= \left\{ \tr \left( p_{i}(\bm{s})f(\bm{s}) \right), \tr \left( q_{i}(\bm{s})f(\bm{s}) \right) : i \in  [t], \bm{s} \in \mathbb{F}_{q^t}^m \setminus \{\bm{s}^*, \bm{s^{\prime}}\} \right\}. \]
To prove this claim, we take the following steps.\newline
Step 1: For $i \in [t-1]$, $\tr \left( z_i f(\bm{s^{\prime}}) \right)$ can be recovered from set $R$ and Equation~\ref{Eq.21.06.11_5}.\newline
Step 2: We have $\displaystyle \tau \in \ker \tr,$ whose $\mathbb{F}_q$-basis is $\{z_1,\ldots, z_{t-1} \}$. Hence, there exist $\alpha_1, \ldots, \alpha_{t-1}$ in $\mathbb{F}_q$ such that $\tau = \alpha_1 z_1 + \ldots + \alpha_{t-1}z_{t-1}$ and
\begin{align*}
\tr \left( \tau f( \bm{s^{\prime}} ) \right) = \sum_{i=1}^{t-1} \alpha _ i \tr \left( z_i f( \bm{s^{\prime}} ) \right).
\end{align*}
Thus, $\tr \left( \tau f( \bm{s^{\prime}} ) \right)$ can be recovered from Step 1. \newline
Step 3: From Step 2 and Equations~\ref{Eq.21.06.11_3} and~\ref{Eq.21.06.11_4}, $\left(\tau z_i f(\bm{s}^*)\right)$ can be recovered for $i \in [t]$. The element $\tau f(\bm{s}^*)$ can then be recovered from the $t$ traces $ \tr(\tau z_i f(\bm{s}^*))$ by \[\tau f(\bm{s}^*) = \tr(\tau z_1 f(\bm{s}^*))z_1^\prime + \ldots + \tr(\tau z_t f(\bm{s}^*))z_t^\prime,\]
where $\{z_1^\prime,\ldots,z_t^\prime\}$ is the dual basis of $\{z_1,\ldots,z_t\}$; see Remark~\ref{21.01.02}.
Thus, $f(\bm{s}^*)$ can be recovered by
\[f(\bm{s}^*) = \tau^{-1}\tr(\tau z_1 f(\bm{s}^*))z_1^\prime + \ldots + \tau^{-1}\tr(\tau z_t f(\bm{s}^*))z_t^\prime.\]
Step 4: From Step 3, and Equations~\ref{Eq.21.06.11_5} and~\ref{Eq.21.06.11_6}, $\tr \left(z_i f(\bm{s}^\prime)\right)$ can be recovered for $i \in [t]$ by the set $R$. Then, similarly to Step 3, the element $f(\bm{s}^\prime)$ can be recovered from the traces $\tr \left(z_i f(\bm{s}^\prime)\right)$. This completes the proof of the claim.

Recall that for $i \in [t]$, we have the following expressions:
\[p_i(\bm{x}) = \tau\frac{\tr \left(z_i (x_j - s^*_j )\right)}{ \left(x_j - s^*_j \right)}
\qquad \text{ and } \qquad
q_i(\bm{x}) = \frac{\tr\left(z_i (x_j - s^{\prime}_j )\right)}{ (x_j - s^{\prime}_j )}.\]

Define the sets \[ \Gamma^* := \mathbb{F}_{q^t} \times \cdots \times \{s^*_j\} \times \cdots \times \mathbb{F}_{q^t} = \{ \left(\gamma_1,\ldots,\gamma_n \right) \in \mathbb{F}_{q^t}^m : \gamma_j = s^*_j \}\]
and
\[ \Gamma^\prime := \mathbb{F}_{q^t} \times \cdots \times \{s^\prime_j\} \times \cdots \times \mathbb{F}_{q^t} = \{ \left(\gamma_1,\ldots,\gamma_n \right) \in \mathbb{F}_{q^t}^m : \gamma_j = s^\prime_j\}.\]
As a consequence of the claim, both erasures $f(\bm{s^{\prime}})$ and $f(\bm{s}^*)$ can be recovered from the set
\begin{align*}
R =& \left\{ \tr \left( p_{i}(\bm{s})f(\bm{s}) \right), \tr \left( q_{i}(\bm{s})f(\bm{s}) \right) : i \in  [t], \bm{s} \in \mathbb{F}_{q^t}^m \setminus \{\bm{s}^*, \bm{s^{\prime}}\} \right\}\\
=& \left\{ \tr \left( p_{i}(\bm{s})f(\bm{s}) \right) : i \in  [t], \bm{s} \in \mathbb{F}_{q^t}^m \setminus \{\bm{s}^*, \bm{s^{\prime}}\} \right\}\\
\cup& \left\{\tr \left( q_{i}(\bm{s})f(\bm{s}) \right) : i \in  [t], \bm{s} \in \mathbb{F}_{q^t}^m \setminus \{\bm{s}^*, \bm{s^{\prime}}\} \right\}
\\
=& \left\{ \tr \left( p_{i}(\bm{s})f(\bm{s}) \right) : i \in  [t], \bm{s} \in \Gamma^* \setminus \{\bm{s}^*, \bm{s^{\prime}}\} \right\}\\
\cup& \left\{ \tr \left( p_{i}(\bm{s})f(\bm{s}) \right) : i \in  [t], \bm{s} \notin \Gamma^* \right\}\\
\cup& \left\{\tr \left( q_{i}(\bm{s})f(\bm{s}) \right) : i \in  [t], \bm{s} \in \Gamma^\prime \setminus \{\bm{s}^*, \bm{s^{\prime}}\} \right\}
\\
\cup& \left\{\tr \left( q_{i}(\bm{s})f(\bm{s}) \right) : i \in  [t], \bm{s} \notin \Gamma^\prime \right\}.
\end{align*}

Observe that 
$$
p_i(\bm{s})= \begin{cases}
   \tau z_i & \textnormal{if }  \bm{s} \in \Gamma^* \\ \ \\
 \frac{ \tau \tr(z_i(s_P - s_j^*))}{(s_P - s_j^*)} & \textnormal{if }   {\bm{s} \notin \Gamma^*}
\end{cases}
\quad \text{and} \quad
q_i(\bm{s})= \begin{cases}
   z_i & \textnormal{if }  \bm{s} \in \Gamma^\prime \\ \ \\
 \frac{ \tr(z_i(s_P - s_j^\prime))}{(s_P - s_j^\prime)} & \textnormal{if }   {\bm{s} \notin \Gamma^\prime}
\end{cases}
$$
where $s_P$ is the first entry of the point $\bm{s} \in \mathbb{F}_{q^t}^m$. Thus, the set $R$ can be written as
\begin{align*}
R =& \left\{ \tr \left( \tau z_i f(\bm{s}) \right) : i \in  [t], \bm{s} \in \Gamma^* \setminus \{\bm{s}^*, \bm{s^{\prime}}\} \right\}\\
\cup& \left\{ \tr \left( \frac{\tau \tr(z_i(s_P - s_j^*))}{(s_P - s_j^*)} f(\bm{s}) \right) : i \in  [t], \bm{s} \notin \Gamma^* \right\}\\
\cup& \left\{\tr \left( z_if(\bm{s}) \right) : i \in  [t], \bm{s} \in \Gamma^\prime \setminus \{\bm{s}^*, \bm{s^{\prime}}\} \right\}
\\
\cup& \left\{\tr \left( \frac{ \tr(z_i(s_P - s_j^\prime))}{(s_P - s_j^\prime)}f(\bm{s}) \right) : i \in  [t], \bm{s} \notin \Gamma^\prime \right\}
\\
=& \left\{ \tr \left( \tau z_i f(\bm{s}) \right) : i \in  [t], \bm{s} \in \Gamma^* \setminus \{\bm{s}^*, \bm{s^{\prime}}\} \right\}\\
\cup& \left\{ \tr(z_i(s_P - s_j^*)) \tr \left( \frac{\tau f(\bm{s})}{(s_P - s_j^*)}  \right) : i \in  [t], \bm{s} \notin \Gamma^* \right\}\\
\cup& \left\{\tr \left( z_if(\bm{s}) \right) : i \in  [t], \bm{s} \in \Gamma^\prime \setminus \{\bm{s}^*, \bm{s^{\prime}}\} \right\}
\\
\cup& \left\{ \tr(z_i(s_P - s_j^\prime)) \tr \left(\frac{f(\bm{s})}{(s_P - s_j^\prime)} \right) : i \in  [t], \bm{s} \notin \Gamma^\prime \right\}.
\end{align*}

Thus, both erasures $f(\bm{s^{\prime}})$ and $f(\bm{s}^*)$ can be recovered by downloading the following information from every element $\bm{s} \in \mathbb{F}_{q^t}^m \setminus\{\bm{s}^*, \bm{s^{\prime}}\}$:
\begin{itemize}
\item if $\bm{s} \in \Gamma^* \setminus\{\bm{s}^*,\bm{s}^\prime\}$,  download $f(\bm{s})$.
\item if $\bm{s} \notin \Gamma^*$, { download } $\displaystyle \tr\left(\frac{ \tau f(\bm{s})}{ s_P - s_j^*}\right)$.
\item if $\bm{s} \in \Gamma^\prime \setminus\{\bm{s}^*,\bm{s}^\prime\}$,  download $f(\bm{s})$.
\item if $\bm{s} \notin \Gamma^\prime$, { download } $\displaystyle \tr\left(\frac{ \tau f(\bm{s})}{ s_P - s_j^\prime}\right)$. 
\end{itemize}
Note $|\Gamma^*| = |\Gamma^\prime|$. Hence, the bandwidth is at most
\begin{align*}
b & = 2\left(t (|\Gamma|-2) + |\mathbb{F}_{q^t}^m\setminus \Gamma| \right) = 2\left(t (q^{(m-1)t}-2) + q^{mt} - q^{(m-1)t}\right)\\
& = 2\left(tq^{(m-1)t} - 2t + q^{mt} - q^{(m-1)t}\right)\\
& = 2\left((t-1)q^{(m-1)t} - 2t + q^{mt}\right)\\
&= 2\left(q^{mt}-2 + (t-1)\left(q^{(m-1)t}-2\right)\right)\\
&= 2\left( n-2 + (t-1)\left(q^{(m-1)t}-2\right)\right),
\end{align*}
which concludes the proof.
\end{proof}

%%%%%%%%%%%%%%%%%%%%%%%%%%%%%%%%%%%%%%%%%%%%%%%%%%%%%%%%%%%%%%%%%%%%%%%%%%%%%%%%
\section{An $\ell$-erasures repair scheme of a Reed-Muller code}
In this section, we adapt the GW scheme of one erasure from Reed-Solomon codes to several erasures on a Reed-Muller code. We give the sketch to prove the case of three erasures. Such a sketch gives the key steps for the general case.

Let $c = (f(\bm{s}_1),\ldots,f(\bm{s}_n))$ be an element of a Reed-Muller code $\grm(\mathbb{F}_{q^t}^m, d)$, where $d\leq m(q^t-1)-q^{t-1}$ and $n=q^{tm}$. Assume that $c$ has the three erasures $f(\bm{s}^{1})$, $f(\bm{s}^{2})$, and $f(\bm{s}^{3})$, where $\bm{s}^i = \left(s_1^i,\ldots,s_n^i \right)$ for $i =1,2,3$. If there is $j \in [n]$ such that $s^{i_1}_j - s^{i_2}_j\in \mathbb{F}_q^*$, for every $i_1 \neq i_2\in[3]$, then there exists a repair scheme for the three erasures with bandwidth at most
\[b = 3\left[ n-3 + (t-1)\left(q^{(m-1)t}-3\right)\right].\]

The sketch of the proof is as follows. Let $\{z_1, \dots, z_{t - 1}\}$ be an $\mathbb{F}_q$-basis for $\ker \tr$ and $z_t$ an element in $\F_{q^t}$ such that $\{z_1, \dots, z_{t - 1}, z_t\}$ is an $\mathbb{F}_q$-basis for $\F_{q^t}.$ Let $\tau_1$ and $\tau_2$ be two elements of $\ker \tr$ that are independent over $\mathbb{F}_q$. % and define $\tau =\displaystyle \frac{\overline{\tau}}{\lambda_{\bm{s^{\prime}}}}.$ %$\Gamma_{\bm{s}^*}^{\bm{s^{\prime}}}.$ Define $\overline{\tau} = \tau\prod_{i=1}^m (s^{\prime}_i - s^*_i ).$
For $i \in [t]$, take
\[p_i(\bm{x}) = \frac{\tr \left(z_i (x_j - s^1_j )\right)}{ \left(x_j - s^1_j \right)},
\qquad
q_i(\bm{x}) = \frac{\tau \tr_1\left(z_i (x_j - s^2_j )\right)}{ (x_j - s^2_j )},\]
\[\text{ and } \qquad
r_i(\bm{x}) = \frac{\tau_2 \tr\left(z_i (x_j - s^3_j )\right)}{ (x_j - s^3_j )}.\]

As $d\leq m(q^t-1)-q^{t-1}$, the polynomials $p_i(\bm{x})$, $q_i(\bm{x})$, and $r_i(\bm{x})$ define elements in the dual code $\grm(\mathbb{F}_{q^t}^m, d)^\perp$~\cite{Lopez2021}. Define the set $S=\{\bm{s}^{1}, \bm{s}^{2}, \bm{s}^{3}\}$. Similarly to the proof of Theorem~\ref{21.06.16}, we obtain the following $3t$ equations:
\begin{align*}
&\sum_{j=1}^3p_{i}(\bm{s}^j)f(\bm{s}^j) =-\sum_{\bm{s} \in \mathbb{F}_{q^t}^m \setminus S} p_{i}(\bm{s})f(\bm{s}), \qquad i\in[t],
\end{align*}
\begin{align*}
&\sum_{j=1}^3q_{i}(\bm{s}^j)f(\bm{s}^j) =-\sum_{\bm{s} \in \mathbb{F}_{q^t}^m \setminus S} q_{i}(\bm{s})f(\bm{s}), \qquad i\in[t],
\end{align*}
and
\begin{align*}
&\sum_{j=1}^3r_{i}(\bm{s}^j)f(\bm{s}^j) =-\sum_{\bm{s} \in \mathbb{F}_{q^t}^m \setminus S} r_{i}(\bm{s})f(\bm{s}), \qquad i\in[t].
\end{align*}
The last equations give rise to
\begin{align*}
\tr \left(z_i f(\bm{s}^1)\right)
&=-\sum_{\bm{s} \in \mathbb{F}_{q^t}^m \setminus S} \tr \left( p_{i}(\bm{s})f(\bm{s})\right),\\
\tr \left(z_t f(\bm{s}^1)\right) + \sum_{j=2,3} \tr\left(z_t\right) \tr \left(f(\bm{s}^j)\right)
&=-\sum_{\bm{s} \in \mathbb{F}_{q^t}^m \setminus S} \tr \left( p_{t}(\bm{s})f(\bm{s})\right),\\
\tr \left(\tau_1 z_i f(\bm{s}^2)\right)
&=-\sum_{\bm{s} \in \mathbb{F}_{q^t}^m \setminus S} \tr \left( q_{i}(\bm{s})f(\bm{s})\right),\\
\sum_{j=1,3} \tr\left(z_t\right) \tr \left(\tau_1 f(\bm{s}^j)\right) + \tr \left(\tau_1 z_t f(\bm{s}^2)\right)
&=-\sum_{\bm{s} \in \mathbb{F}_{q^t}^m \setminus S} \tr \left( q_{t}(\bm{s})f(\bm{s})\right),\\
\tr \left(\tau_2 z_i f(\bm{s}^3)\right)
&=-\sum_{\bm{s} \in \mathbb{F}_{q^t}^m \setminus S} \tr \left( r_{i}(\bm{s})f(\bm{s})\right),\\
\sum_{j=1,2} \tr\left(z_t\right) \tr \left(\tau_2 f(\bm{s}^j)\right) + \tr \left(\tau_2 z_t f(\bm{s}^3)\right)
&=-\sum_{\bm{s} \in \mathbb{F}_{q^t}^m \setminus S} \tr \left( r_{t}(\bm{s})f(\bm{s})\right),
\end{align*}
for $i\in[t-1]$.

In a similar way to the proof of Theorem~\ref{21.06.16}, the elements $f(\bm{s}^1)$, $f(\bm{s}^2)$, and $f(\bm{s}^3)$ can be recovered from the set \[ R:= \left\{ \tr \left( p_{i}(\bm{s})f(\bm{s}) \right), \tr \left( q_{i}(\bm{s})f(\bm{s}) \right), \tr \left( r_{i}(\bm{s})f(\bm{s}) \right)  : i \in  [t], \bm{s} \in \mathbb{F}_{q^t}^m \setminus S \right\}.\]

For $i\in[3]$, define the sets \[ \Gamma^i := \mathbb{F}_{q^t} \times \cdots \times \{s^*_i\} \times \cdots \times \mathbb{F}_{q^t} = \{ \left(\gamma_1,\ldots,\gamma_n \right) \in \mathbb{F}_{q^t}^m : \gamma_j = s^i_j \}.\]
As a consequence, following the proof of Theorem~\ref{21.06.16}, the three erasures $f(\bm{s}^1)$, $f(\bm{s}^2)$, and $f(\bm{s}^3)$ can be recovered by downloading the following information for every element $\bm{s} \in \mathbb{F}_{q^t}^m \setminus S$ and $i=1,2,3$:
\begin{itemize}
\item if $\bm{s} \in \Gamma^i \setminus S$,  download $f(\bm{s})$.
\item if $\bm{s} \notin \Gamma^i$, download $\displaystyle \tr\left(\frac{ \tau f(\bm{s})}{ s_P - s_j^i}\right)$.
\end{itemize}
As $|\Gamma^1| = |\Gamma^2| = |\Gamma^3|$, the bandwidth is at most
\begin{align*}
b & = 3\left(t (|\Gamma|-3) + |\mathbb{F}_{q^t}^m\setminus \Gamma| \right)\\
&= 3\left(t (q^{(m-1)t}-3) + q^{mt} - q^{(m-1)t}\right)\\
& = 3\left(tq^{(m-1)t} - 3t + q^{mt} - q^{(m-1)t}\right)\\
& = 3\left((t-1)q^{(m-1)t} - 3t + q^{mt}\right)\\
&= 3\left(q^{mt}-3 + (t-1)\left(q^{(m-1)t}-3\right)\right)\\
&= 3\left( n-3 + (t-1)\left(q^{(m-1)t}-3\right)\right),
\end{align*}
which concludes the sketch of the proof.

In general, the case of $\ell$ erasures for the Reed-Muller code can be stated in the following way.
\begin{theorem}
Let $c = (f(\bm{s}_1),\ldots,f(\bm{s}_n)) \in \grm(\mathbb{F}_{q^t}^m, d)$, where $\grm(\mathbb{F}_{q^t}^m, d)$
is a Reed-Muller code 
with
$d\leq m(q^t-1)-q^{t-1}$ and $n=q^{tm}$. Assume that $c$ has  $\ell$ erasures $f(\bm{s}^{1}), \ldots, f(\bm{s}^{\ell})$, where $\bm{s}^i = \left(s_1^i,\ldots,s_n^i \right)$ for $i =1,\ldots,\ell$. If there is $j \in [n]$ such that $s^{i_1}_j - s^{i_2}_j\in \mathbb{F}_q^*$, for every $i_1 \neq i_2\in[\ell]$, then there exists a repair scheme for the $\ell$ erasures with bandwidth at most
\[b = \ell\left[ n-\ell + (t-1)\left(q^{(m-1)t}-\ell\right)\right].\]
\end{theorem}
\section{Summary/Conclusion}
A distributed storage system stores data across multiple
nodes, with the primary objective of enabling efficient data recovery even in the event of node failures. The main goal of an exact repair scheme is to recover the data from a failed node by accessing and downloading information from the rest of the nodes. In these notes, we extended the exact repair scheme developed in~\cite{GW} from Reed-Solomon codes to Reed-Muller codes with several erasures that satisfy certain conditions.

\end{document}